\documentclass[12pt]{article}

\usepackage{natbib}
\usepackage{enumerate}
\usepackage{graphicx}
\graphicspath{ {images/} }
\usepackage{float}
\usepackage{caption}
\usepackage{subcaption}
\usepackage{amsthm}
\usepackage{amsfonts}
\usepackage{bm}
\usepackage{changepage}
\usepackage[margin=0.8in]{geometry}
\newtheorem{theorem}{Theorem}[section]

\newtheorem{lemma}[theorem]{Lemma}

\title{On the Complexity of 2-Player Packing Games}
\author{%
  Oliver Korten%
    \thanks{Tufts University, Medford, MA, USA}}
\begin{document}
\maketitle

\begin{center}
    \section*{Abstract}
\end{center}
\begin{adjustwidth}{10pt}{10pt}
We analyze the computational complexity of two 2-player games involving packing objects into a box. In the first game, players alternate drawing polycubes from a shared pile and placing them into an initially empty box in any available location; the first player who can't place another piece loses. In the second game, there is a fixed sequence of polycubes, and on a player's turn they drop the next piece in through the top of the box, after which it falls until it hits a previously placed piece (as in Tetris); the first player who can't place the next piece loses. We prove that in both games, deciding the outcome under perfect play is PSPACE-complete.
\end{adjustwidth}

\section{Introduction}\label{Intro}
Many NP-complete problems take the form of 1-player packing puzzles: there are a set of ``pieces'' and a board, and the player's goal is to pack all of the pieces into the board, subject to the restriction that all of them ``fit''. In polyomino packing, proved hard in \cite{dem}, the pieces are polyominoes and the board is a rectangle, and ``fitting'' in this context means being placed with no overlap. Tetris, first proved hard in \cite{tetris} and analyzed further in \cite{totaltetris}, uses the same kinds of pieces, but in this puzzle the pieces must be dropped in from the top of a box and fall until they hit other pieces, and the order of the available pieces is a fixed part of the input. Since Robertson and Munro published their analysis of Instant Insanity \cite{insanity}, it has been known that PSPACE-complete problems often take the form of 2-player generalizations of 1-player puzzles. Thus, a natural question to ask is the complexity of 2-player packing games, where players take turns placing pieces into an arrangement, and the first player who cannot extend the packing loses.

In this paper we analyze two 2-player games involving packing physical pieces into a box. In Section~\ref{polypacking}, we introduce 2-Player Polycube Packing, a 2-player generalization of the known NP-hard polyomino packing problem. In 2-Player Polycube Packing, players take turns placing polycubes into an initially empty box, and the first player who can't place another polycube loses. We prove that deciding the winner of this game under perfect play is PSPACE-complete.

In Section~\ref{tetris}, we introduce 2-Player 3D $n$-tris, a 2-player generalization of Tetris. In this game there is a set sequence of polycubes, and players take turns placing the next polycube in the sequence into an initially empty box by continuous downward motion until it rests on top of a previously placed piece. The first player who can't place a piece loses. Once again, we show that deciding a winner under optimal play in this game is PSPACE-complete. Finally, in Section~\ref{open}, we conclude with some open problems.

For both games defined below, membership in PSPACE is justified by the fact that these games are perfect information and end after a polynomial number of moves \cite{sch}. In order to establish PSPACE-hardness of our games, we will reduce from a game called Node Kayles. In this game, the input is an undirected graph $G$. Players take turns marking vertices in $G$, with the restriction that marked vertices cannot be adjacent. The first player who cannot make a move loses. Deciding a winner in this game was shown to be PSPACE-complete in \cite{sch}.

\section{2-Player Polycube Packing}\label{polypacking}
2-Player Polycube Packing is defined as follows: the input is an $N \times M \times K$ box and a set of polycubes. Two players take turns placing any remaining polycube into any available grid-aligned location in the box using translations, rotations, and (optionally) reflections. The first player who can't place any more polycubes loses.
\begin{theorem}
It is PSPACE-complete to determine the winner of a Polycube Packing game from an initially empty $N \times M \times 3$ box.
\end{theorem}
We will prove PSPACE-hardness of Polycube Packing by reducing from the problem of deciding a winner in Node Kayles. The proof holds regardless of if we allow rotations only, or rotations and reflections. The proof also uses only simply-connected polycubes.

\subsection{Overview of Reduction}
Given an instance of Node Kayles in the form of a graph $G$, we will construct a set of polycubes and a complex initial board state such that each remaining polycube can only be placed in exactly one spot on the board. Further, we will construct this initial state such that if one were to superimpose all of these remaining polycubes onto the board in their unique spots, the intersection graph of these polycubes would be isomorphic to $G$. In this way, each playable polycube plays the role of a vertex in Node Kayles, and playing a polycube (by putting it in its unique spot) removes from future play exactly those polycubes which represent the neighbors of its respective vertex in $G$. Next, it is shown that this complex initial state can actually be arrived at after the placement of a single polycube, which we call the $G$-mold. Finally, we prove that by adding the $G$-mold and a few other pieces into the set of playable polycubes, we can get this reduction to work even from an initially empty board state.

\textit{\textbf{Note:}} In every figure, upwards is the positive $y$-direction, right is the positive $x$-direction, and out of the page is the positive $z$-direction, unless otherwise specified by a 3-dimensional axis icon. In this icon, the arrows pointing straight up and straight right indicate dimensions oriented upwards and rightwards respectively, and the arrow bisecting these two orthogonal arrows indicates the dimension oriented into the page. An example appears in Figure~\ref{fig:crossover}.

\subsection{Hardness from a Complex Initial Board State}

\subsubsection{Wiring Diagram}
Given a graph $G = (V, E)$, where $V = \{v_1, ..., v_n\}$ and $E = \{e_1, ..., e_m\}$, first we will assign to each $v_i$ a horizontal line segment with $x$-coordinate $4i$, spanning the $y$-interval $[0, 2m]$. Now, for each edge $e_i$ which connects $v_j$ to $v_k$, draw a horizontal line segment with endpoints $(4j, 2i)$ and $(4k, 2i)$. We will refer to this 2-dimensional layout of $G$ as its wiring diagram. Figure~\ref{fig:wiring diagram} shows the wiring diagram of $H = (V', E')$ where\\
\centerline{$V' = \{v_1, v_2, v_3, v_4, v_5, v_6, v_7\}$} 
\centerline{$E' = \{v_1v_4, v_6v_7, v_2v_4, v_4v_6, v_3v_5, v_1v_3, v_1v_4, v_5v_7\}$}
\begin{figure}[H]
	\centering
	\includegraphics[scale=0.5]{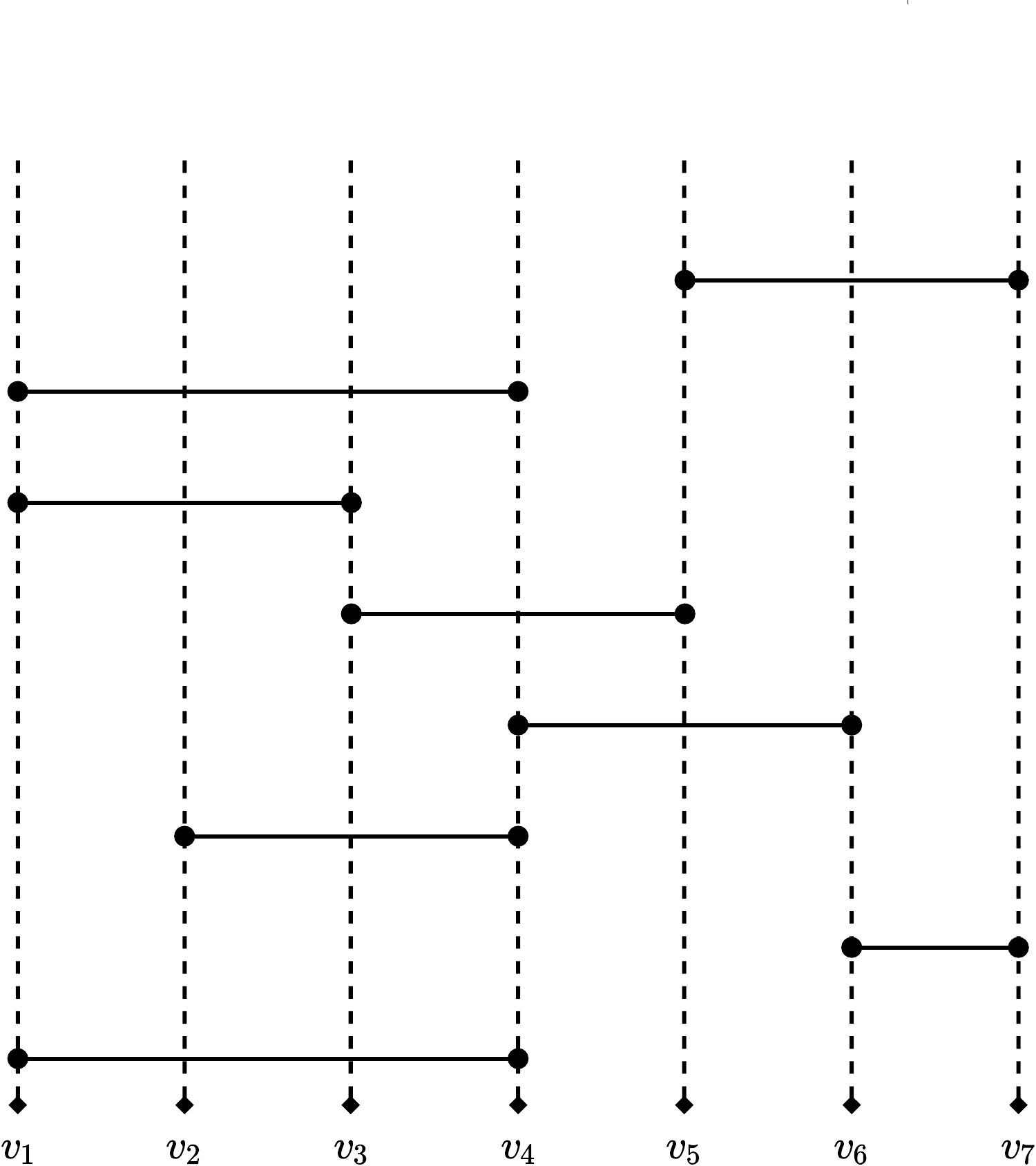}
	\caption{Wiring diagram of H}
	\label{fig:wiring diagram}
\end{figure}
\subsubsection{Cavity}
We will now construct a polycubal structure which mimics the wiring diagram of $G$. The first step is to build a vertical $1 \times 1 \times k$ column for each vertex, where $k$ is the height of the vertical line corresponding to that vertex in the wiring diagram (dotted lines in Figure~\ref{fig:wiring diagram}). Now place these vertex columns with their left boundaries along those corresponding vertical lines in the plane. 

Now, for each edge $e_j$, create a horizontal $1 \times 1 \times k$ column where $k$ is the length of the horizontal segment corresponding to $e_j$ in the wiring diagram, and place this at the same $y$ coordinate as the corresponding segment in the wiring diagram. Say this edge connects $v_i$ to $v_h$, $i < h$. So we will glue the ends of it to the vertical columns for $v_i$ and $v_h$ by laying the two vertical columns and the horizontal column out in the plane as described, and then taking the union of these three polycubes. We then make the following modification to this construction: for every $v_k$ where $i < k < h$, bend $e_j$ to cross over $v_k$ via the construction shown in Figure~\ref{fig:crossover}.

\begin{figure}[H]
	\center
	\includegraphics[scale=0.4]{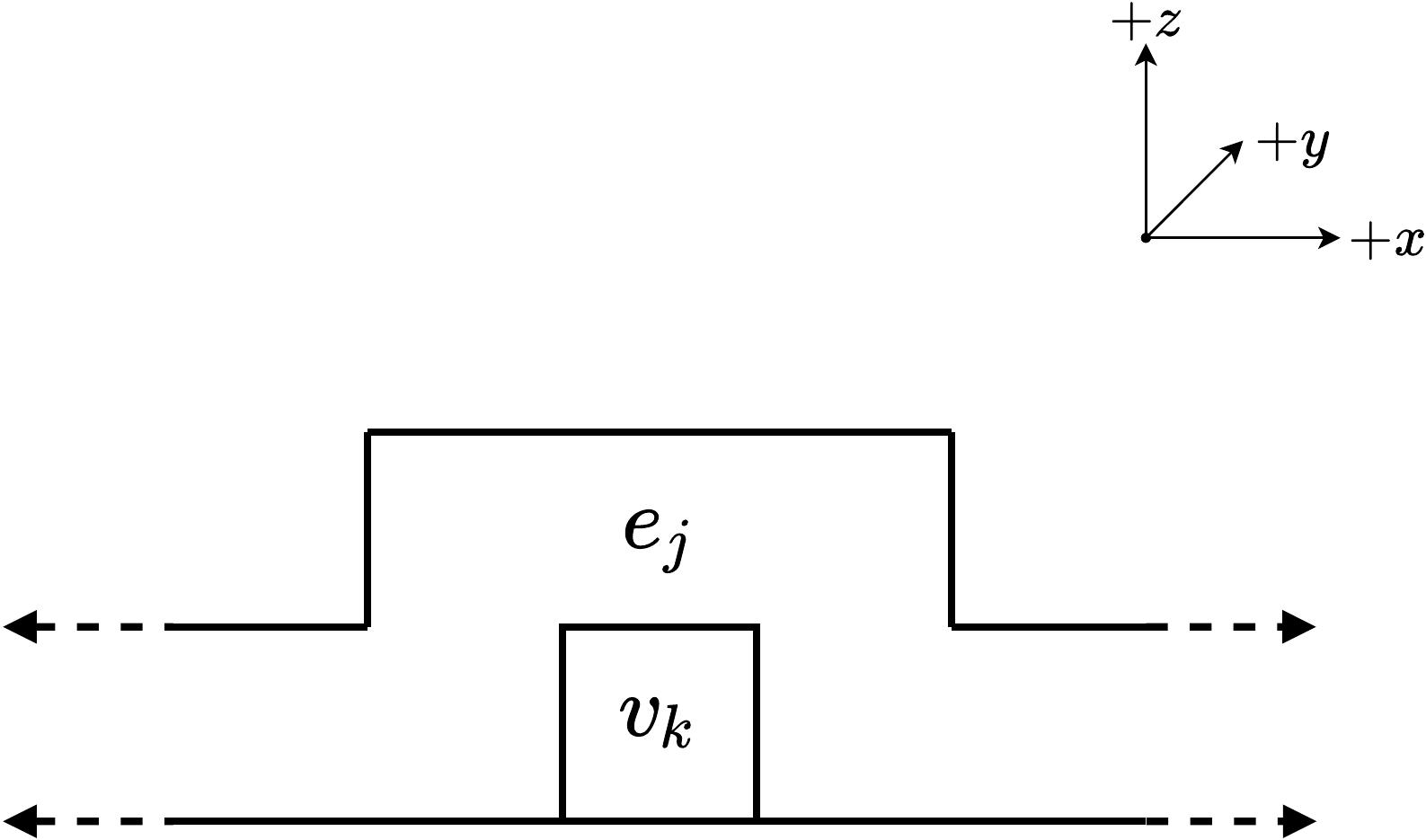}
	\caption{Edge crossing over vertex}
	\label{fig:crossover}
\end{figure}
This part of the contruction will be refered to as the ``wiring region'' of the cavity. The wiring region for the example graph $H$ is shown in Figure~\ref{fig:wiring region}.
\begin{figure}[H]
	\center
	\includegraphics[scale=0.5]{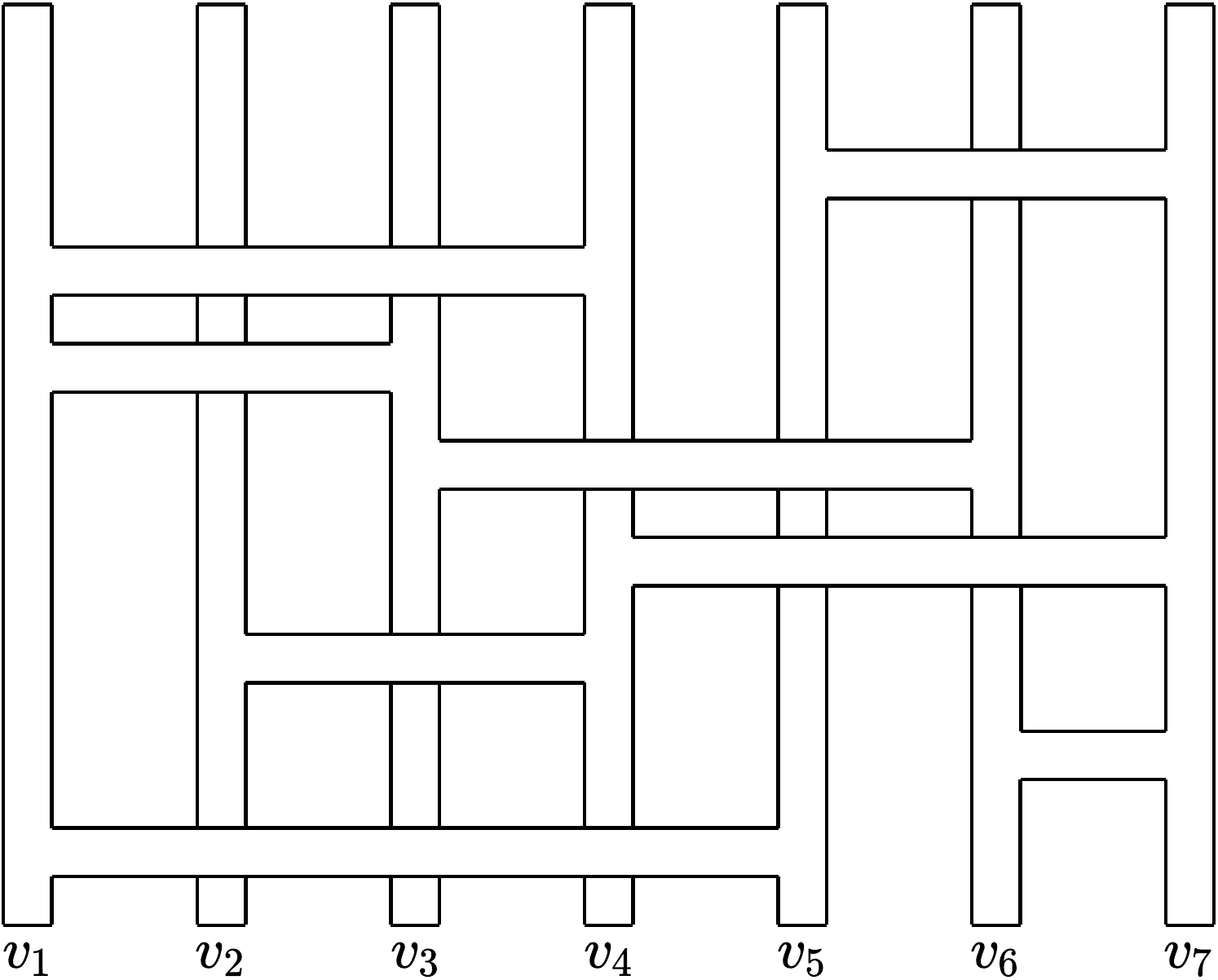}
	\caption{Wiring Region for $H$}
	\label{fig:wiring region}
\end{figure}
Next, we will add ``support beams" to each of the vertex columns as follows. First, extend the tops of the vertex columns by $2n$ units. Now, for vertex $v_i$, add two horizontal beams $2i$ units above the bottom of the extension, one reaching from the vertex column of $v_i$ to one unit past the left limit of the wiring region, and the other reaching to one unit past the right limit. These horizontal pieces cross over other vertex columns in the same way edges cross over vertex columns (as shown in Figure~\ref{fig:crossover}). The support beams will be used to ensure that playable pieces must be placed with the correct $x$-shift.

The final step will be to create ``keys'' at the bottom of each vertex column. To do this, extend each vertex column at the bottom by $2n$ units, and attach a single cube hanging off of the left-hand side at the top of this extension. Now, we have completed construction of what we will refer to as the $G$-cavity. Figure~\ref{fig:full cavity} shows the complete cavity construction for the example graph $H$ including the wiring region, keys, and support beams.
\begin{figure}[H]
	\center
	\includegraphics[scale=0.25]{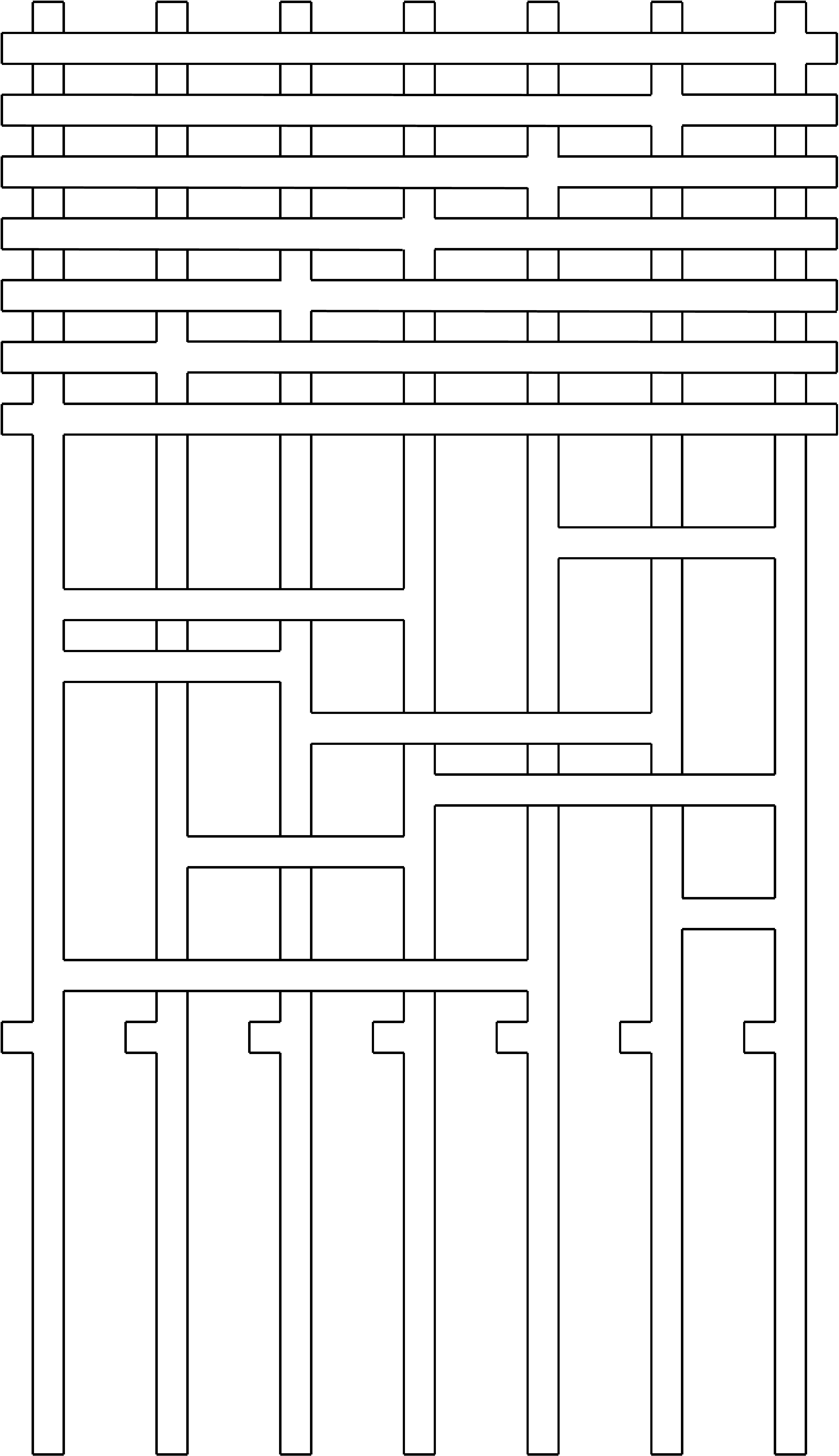}
	\caption{Full Cavity Construction for $H$}
	\label{fig:full cavity}
\end{figure}

\subsubsection{Mold}
The purpose of the $G$-cavity is to serve as the effective board, constraining which of the remaining pieces can be played and where. In order to achieve this, we will create a piece such that when it is placed in an $N \times M \times 3$ bounding box, it leaves the cavity as the only remaining open space. First, create an $N \times M \times 3$ bounding box which is exactly large enough the fit the $G$-cavity; take $M$ to be the height of the $G$-cavity in the $y$ direction and $N$ to be its width in the $x$ direction. By construction, the height of the $G$-cavity in the $z$ direction is 2; it's flat everywhere except at the crossovers which bulge upwards one unit in the z direction. We will now place the $G$-cavity in the bounding box such that the crossover ``bulges'' face upwards in the $z$ direction, and the bottom of the cavity is flush with the bottom of the box. Now, take the $G$-mold to consist of every voxel in the bounding box which is not occupied by the cavity. Figure 6 shows this construction one $y$-valued layer at a time. It is straightforward to see that since each $y$-valued layer of this piece is simply connected and each layer connects to the next by proper face-to-face contact, the entire construction is simply connected.
\begin{figure}[H]
	\centering
	\begin{subfigure}{.5\textwidth}
		\centering
		\includegraphics[width=.8\linewidth]{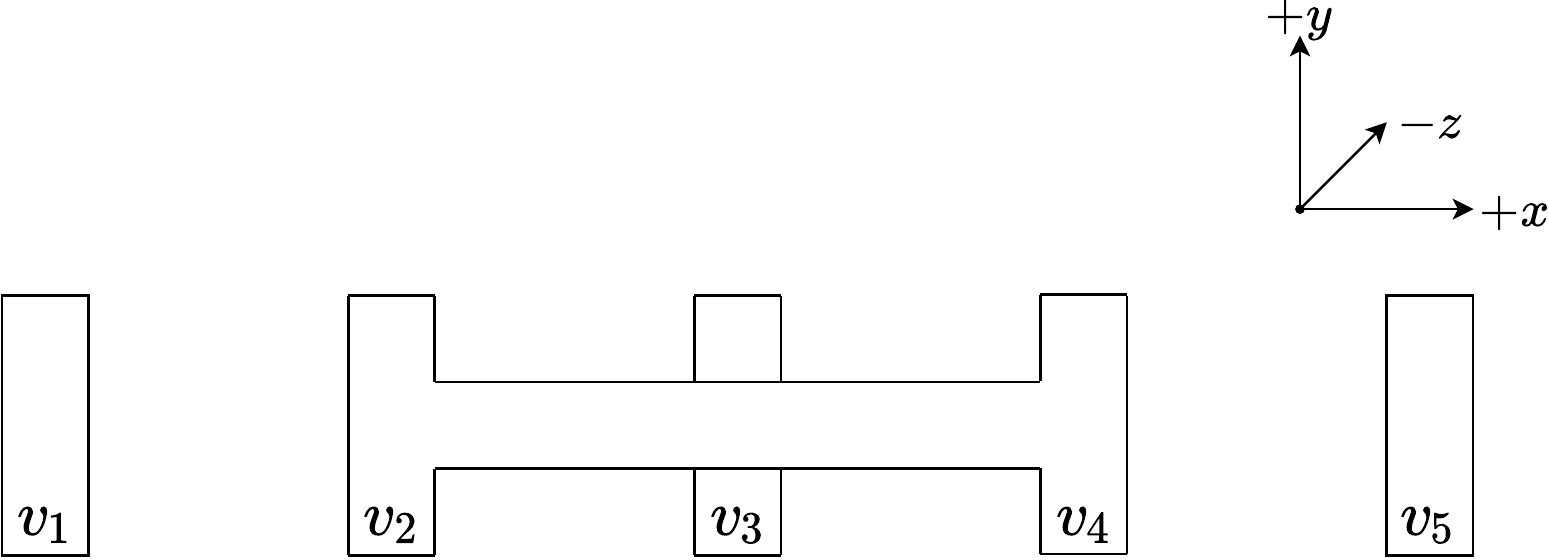}
		\caption{Top-down View of Cavity}
		\label{fig:mold 1}
	\end{subfigure}%
	\begin{subfigure}{.5\textwidth}
		\centering
		\includegraphics[width=.8\linewidth]{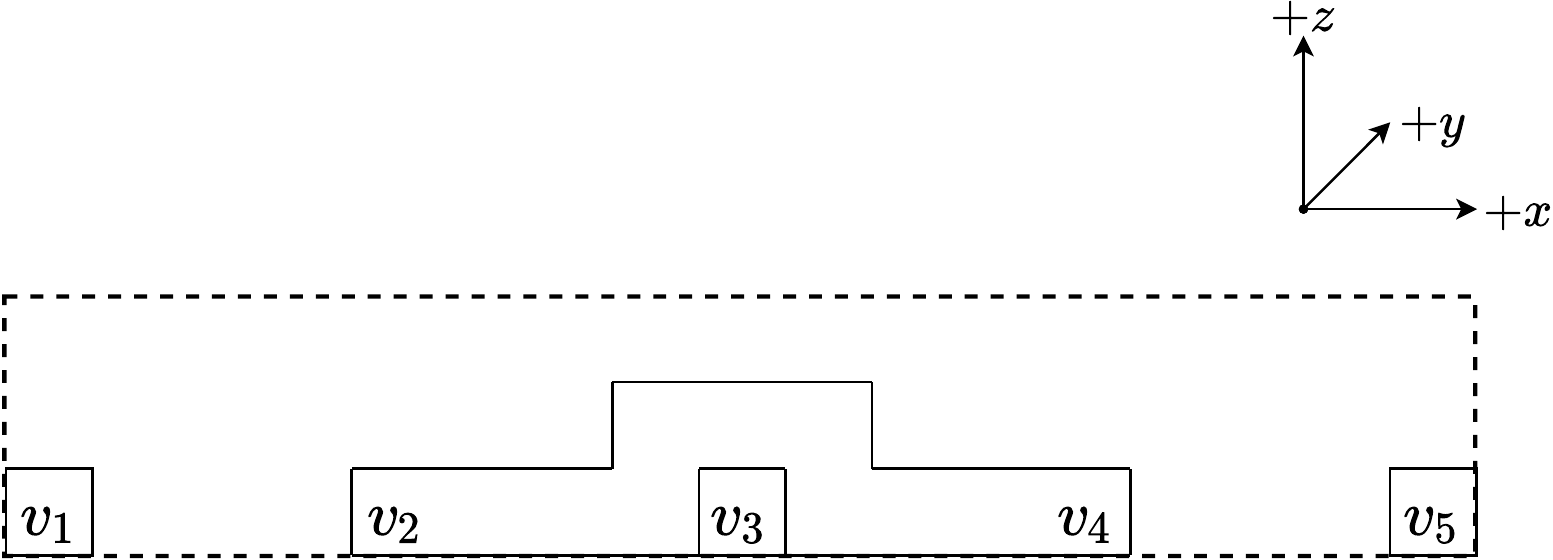}
		\caption{Cross-sectional View of Cavity}
		\label{fig:mold 2}
	\end{subfigure}
	\begin{subfigure}{.5\textwidth}
		\centering
		\includegraphics[width=.8\linewidth]{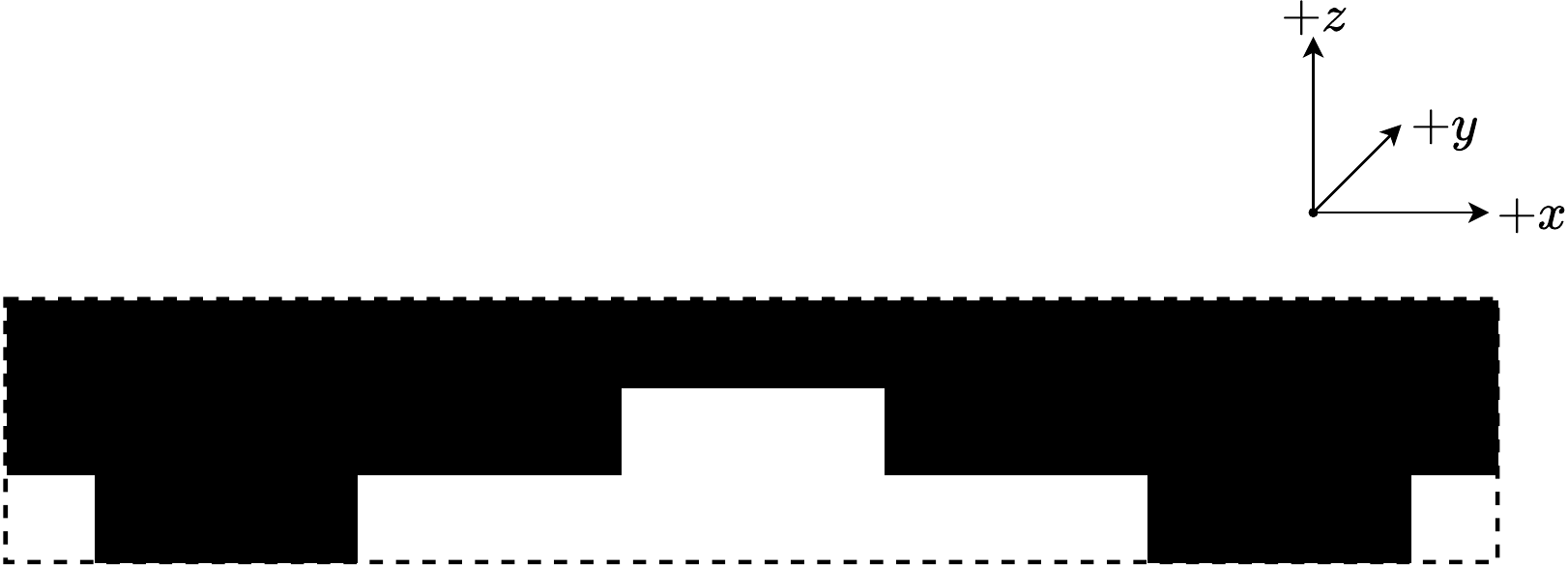}
		\caption{Mold Construction}
		\label{fig:mold 3}
	\end{subfigure}%
\caption{Constructing the mold at a fixed $y$-coordinate. Dotted line shows box boundary, and shaded region shows the mold.}
\label{fig:mold}
\end{figure}
\subsubsection{Vertex Pieces}
The final step in the reduction is to create pieces corresponding to each vertex. To do this, for each $v_i$, take the part of the $G$-cavity construction corresponding to $v_i$ and all of its incident edges, along with its support beams and its key. An example construction for the vertex piece for $v_4$ in $H$ is shown in Figure~\ref{fig:piece}.
\begin{figure}[H]
	\center
	\includegraphics[scale=0.1]{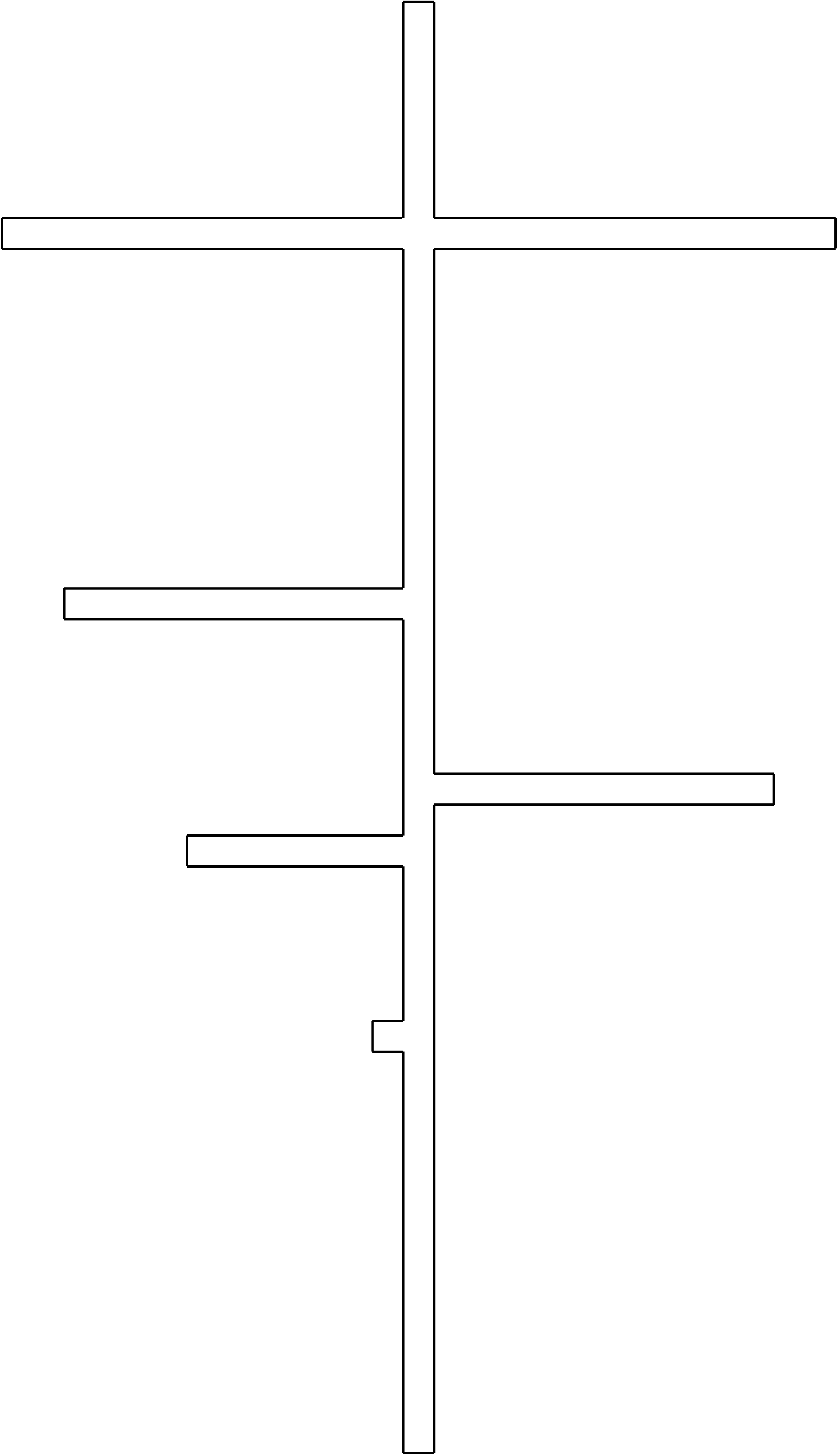}
	\caption{Vertex Piece for $v_4$ in $H$. Note that this piece still contains the crossovers on its edges/support beams despite appearing flat in this figure.}
	\label{fig:piece}
\end{figure}

\begin{lemma}\label{placement}
Once the $G$-mold has been placed in the $N \times M \times 3$ bounding box, each vertex piece fits only into its corresponding location in the cavity.
\end{lemma}
\begin{proof}
Say that the $G$-mold has been placed in the bounding box such that the cavity has the orientation shown in Figure~\ref{fig:full cavity}. This is without loss of generality since the $G$-mold spans every dimension of the bounding box and thus has only one valid placement up to rotation and reflection. First, note that by construction, the $G$-mold occupies every cube in the $z = 3$ plane. Further, the $z = 2$ plane is almost entirely full except for a set of disjoint $1 \times 1 \times 3$ tubes which are the result of edge-vertex crossovers.
Since every vertex piece has a vertex column which is the length of the entire box in the $y$-dimension, and the $y$-dimension is the longest dimension of the box by construction, this column must be placed spanning the $y$-dimension of the box. Moreover, since the $z = 3$ and $z = 2$ planes are almost entirely filled by the $G$-mold, the vertex piece must be placed with its column spanning the $y$-dimension of the $z=1$ plane.

Next, we show that every vertex piece must be placed with its key section at the $-y$ region of the board and its support beam section at the $+y$ region. Since the support beam region is the same height as the key region, if a vertex piece were placed upside-down, its support beams would be in the key section of the $G$-mold. But since in the key section of the $G$-mold, the entire $z=2$ plane is filled (no crossovers occur), and since at least one of the support beams on every vertex piece crosses over some other vertex column, the cross-over section of the piece's support beam would intersect the filled $z=2$ plane in the key region of the $G$-mold. So every vertex piece must be placed with its vertex column in the $z=1$ plane such that its key faces the bottom ($-y$) end, and its support beams faces the top end. Since every vertex piece has at least one support beam longer than 3 units, every piece must be placed with its support beams spanning left and right across the $x$-dimension. Further, the edge crossovers must be facing upwards since if they faced down they would intersect the bottom of the bounding box. Thus, if we only allow rotations, the proof is complete; the vertex column must sit in the $z=1$ plane spanning the $y$-dimension, the key side must face $-y$ side of the box, and since the edge crossovers must face upwards in the $z$ direction, the fact that the support beams span the $x$-dimension fixes the position of the piece along the $x$-axis.

If reflections are allowed, the only other option is for the piece to be reflected across the $x=0$ plane from its ``correct'' position (the position we just proved it is forced into if reflections aren't used). However, if a vertex $v_i$ is placed in such a way, then it occupies the vertex column for vertex $v_{n-i+1}$ in the $G$-cavity. Since it has been reflected, its key peg now protrudes from the right of its vertex column. But in the key region, the column immediately to right of every vertex column is completely filled in the $z=1$ plane by the $G$-mold, so the key peg would intersect the $G$-mold. So such a placement is not possible.
\end{proof}

\subsubsection{Play}
\begin{lemma}\label{hard with mold}
    Once the $G$-mold has been placed, the next player to move has a winning strategy in this instance of Polycube Packing if and only if the first player has a winning strategy in Node Kayles on $G$.
\end{lemma}
\begin{proof}
Imagine that the $G$-mold has already been placed in the $N \times M \times 3$ box, and the remaining pieces are exactly the vertex pieces for each $v_i$. Due to Lemma~\ref{placement}, each $v_i$ can be played only in its corresponding location in the $G$-mold. Once some $v_i$ piece is played in its unique spot, the pieces which it prevents from being played in the future are exactly the pieces corresponding to the vertices $v_i$ is adjacent to in $G$. This is because the $v_i$ piece is constructed to contain the edge pieces for all of its edges, so if $v_i$ has some neighbor $v_j$, the vertex pieces for $v_i$ and for $v_j$ will both occupy the cavity for the edge $v_iv_j$, so playing either one prevents the other from being played. By construction, the forced positions of the vertex pieces are otherwise completely disjoint, so any two pieces corresponding to non-adjacent vertices do not prevent one another from being played. Thus, this game plays exactly like Node Kayles, so winning positions in one game correspond precisely to winning positions in the other. 
\end{proof}

\begin{figure}[H]
	\center
	\includegraphics[scale=1]{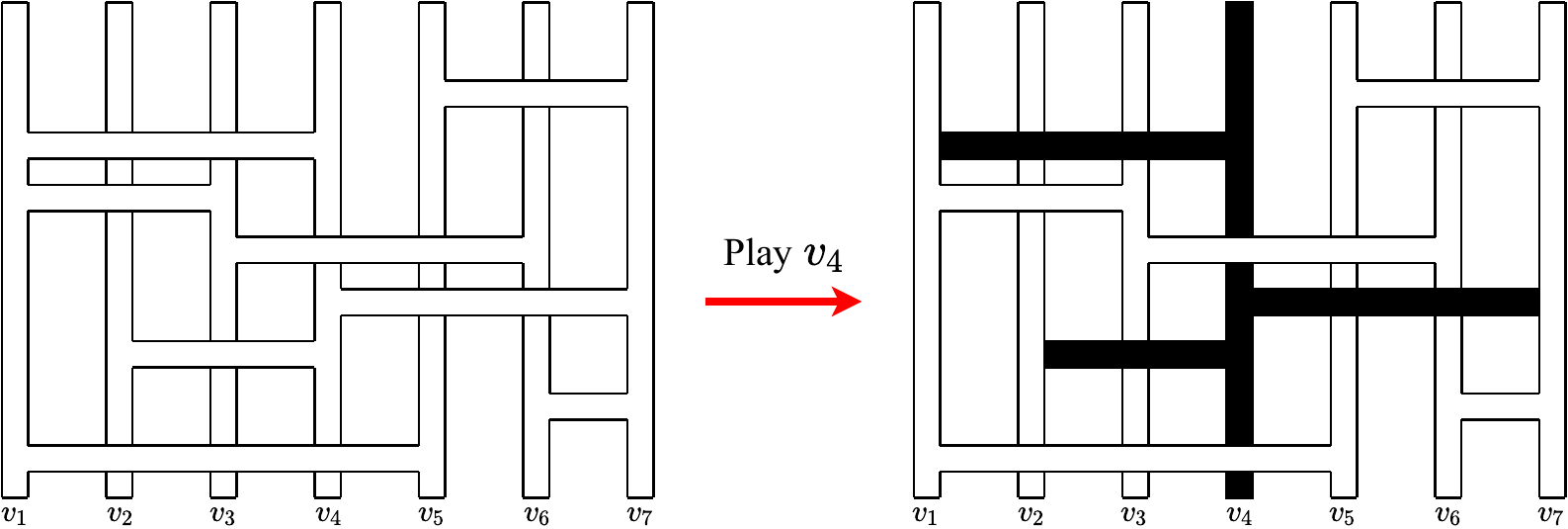}
	\caption{Playing a Vertex Piece}
	\label{fig:play}
\end{figure}

\subsection{Hardness From an Initially Empty Board} The above proof shows that it is PSPACE-hard to determine if a player can win once a single piece, the $G$-mold, has been placed. We will now show that it is also PSPACE-hard to determine a winner from an initially empty $N \times M \times 3$ box. Given an instance of Kayles just as before, we will use the same reduction to create an instance of Polycube Packing, except now we reverse the roles of who plays first, and add some pieces. 

\subsubsection{New Pieces} First, we will add the $G$-mold into the set of playable pieces. Next, for each vertex piece $p_i$, we will add two new pieces, $q_i^-$ and $q_i^+$. The purpose is to create these pieces such that if $p_i$ is played first, either $q_i^-$ or $q_i^+$ can be played second and the game immediately ends, and conversely if $q_i^-$ or $q_i^+$ is played first, $p_i$ can be played second and the game immediately ends. If we can construct $q_i^-$ and $q_i^+$ such that this is accomplished, than the first player will be forced to play the $G$-mold piece in order to avoid immediate defeat, and if we can ensure that playing the $G$-mold makes all $q_i^-$ and $q_i^+$ unplayable, then the remainder of the game will play out exactly as in the original reduction.

Because of the use of support beams which cause each vertex piece to span the $x$ and $y$ dimensions of the bounding box, we know that each vertex piece $p_i$ can only be played in the empty bounding box in one of two grid positions, modulo a reflection/rotation: either it is placed with its vertex column in the $z=1$ plane and its edge crossovers rising into the $z=2$ plane, or it is placed with its vertex column in the $z=2$ plane and its edge crossovers in the $z=3$ plane We will call the first such position ``normal'' position, and the second ``raised'' position. We can ignore reflections and $180^{\circ}$ rotations since if we can construct the $q_i^-$ and $q_i^+$ to fill the rest of the box once $p_i$ is played in a given orientation, then if $p_i$ is reflected or rotated by $180^{\circ}$, $q_i^-$ and $q_i^+$ can be reflected/rotated in the same way.

Imagine that $p_i$ is placed in normal position. We will construct $q_i^-$ to (almost) fill the rest of the box. The construction is very similar to that of the $G$-mold. First, fill the $z=3$ plane completely, then fill every voxel in the $z=2$ plane not occupied by an edge crossover, and finally, fill every voxel in the $z=1$ layer which is unoccupied and is not below an edge crossover. By not filling the gaps under the edge crossovers we ensure that this construction is simply connected.

Next, imagine that $p_i$ is placed in raised position. We construct $q_i^+$ as follows. First, fill the entire $z=1$ plane. Now, fill every unoccupied voxel in the $z=2$ plane. Finally, fill every voxel in the $z=3$ plane which is not occupied and which is not directly above an edge, support beam, or key stub. Again, by not filling areas directly above edges/beams/keys we ensure simply connected pieces. These new piece constructions are shown in Figure~\ref{fig:q pieces}.

\begin{figure}[H]
	\centering
	\begin{subfigure}{.5\textwidth}
		\centering
		\includegraphics[width=.8\linewidth]{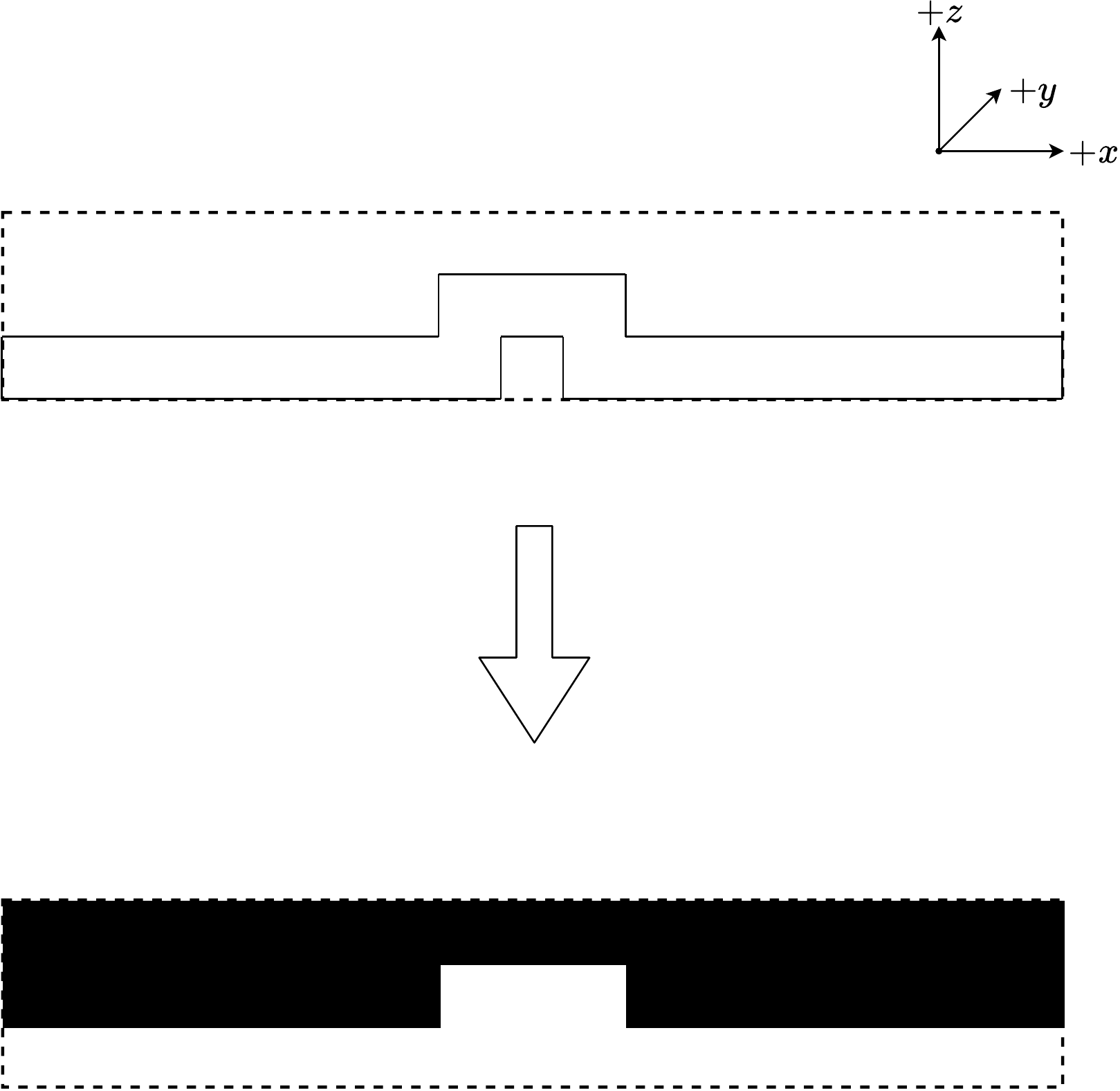}
		\caption{$q_i^-$ Construction}
		\label{fig:normal}
	\end{subfigure}%
	\begin{subfigure}{.5\textwidth}
		\centering
		\includegraphics[width=.8\linewidth]{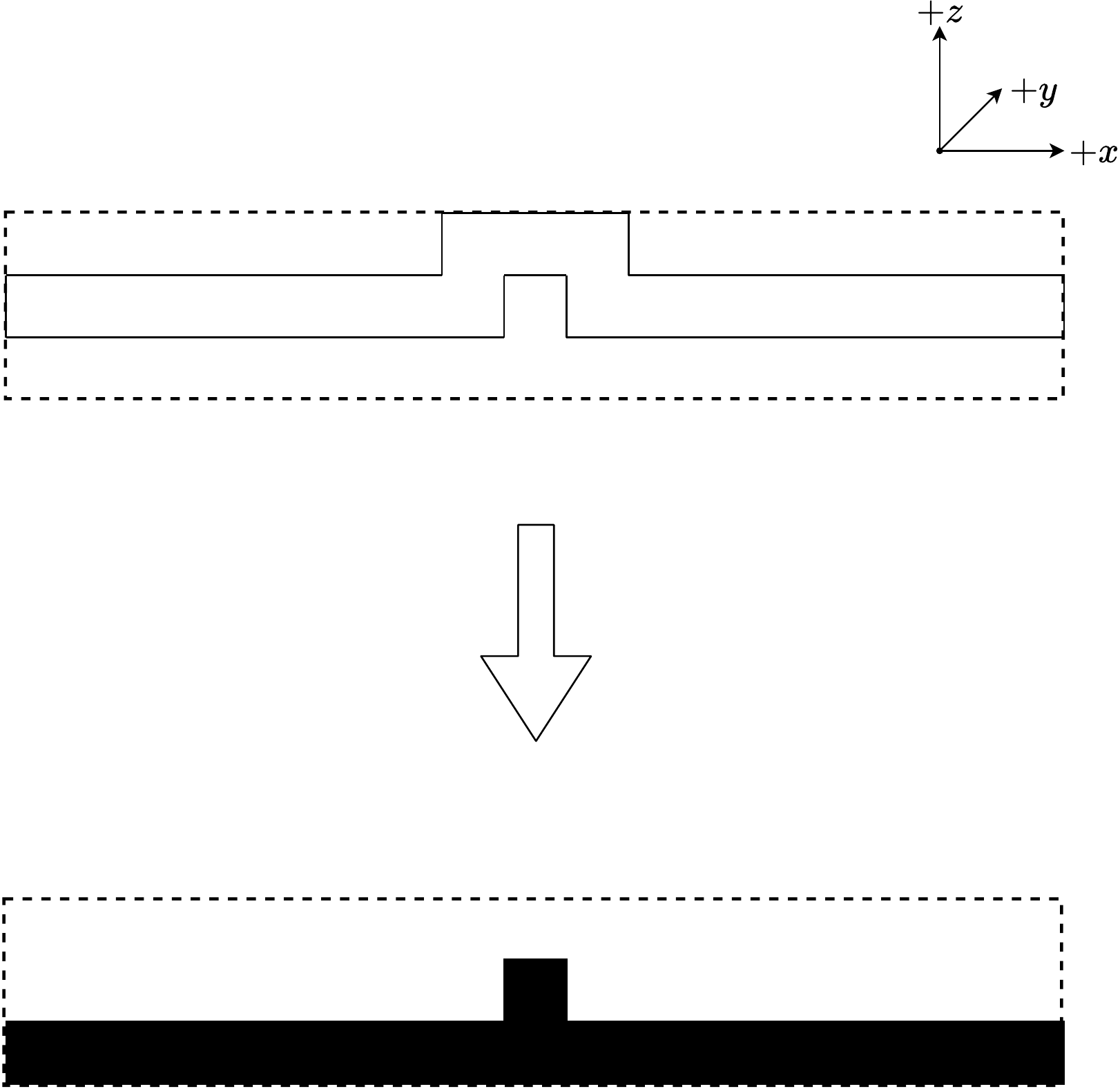}
		\caption{$q_i^+$ Construction}
		\label{fig:raised}
	\end{subfigure}
	\caption{Construction of $q_i^-$ and $q_i^+$. Top figures show the cross-sectional view of a vertex piece in normal (raised) position. Bottom figures show construction of $q_i^-$ ($q_i^+$) based on the position of the vertex piece. Dotted line shows box boundary, and shaded region shows $q_i^-$ ($q_i^+$) piece.}
	\label{fig:q pieces}
\end{figure}

\begin{lemma}\label{lemma normal}
	If $p_i$ is placed in the empty bounding box in normal position, $q_i^-$ can be placed subsequently, after which no other pieces can fit into the board.
\end{lemma}
\begin{proof}
	$q_i^-$ was constructed around $p_i$'s placement in the box in normal position, so we know once $p_i$ is placed in normal position $q_i^-$ will fit. Once they are placed together, $q_i^-$ fills the entire $z=3$ plane, and every voxel in the $z=2$ plane is filled either by an edge crossover or by $q_i^-$ (this is how we constructed $q_i^-$). So if any piece were to fit, it would have to fit entirely in the $z=1$ plane. But by construction, every piece has thickness at least 2 in every dimension so no piece will fit.
\end{proof}

\begin{lemma}\label{lemma raised}
	If $p_i$ is placed in the empty bounding box in raised position, $q_i^+$ can be placed subsequently, after which no other pieces can fit into the board.
\end{lemma}
\begin{proof}
	The proof is almost identical to that of Lemma~\ref{lemma normal}. Again, $q_i^+$ was constructed around $p_i$'s placement in the box in raised position, so we know once $p_i$ is placed in raised position $q_i^+$ will fit. Once they are placed together, $q_i^+$ fills the entire $z=1$ plane, and every voxel in the $z=2$ plane is filled either by $q_i^+$ or by $p_i$. Again, this is because we constructed $q_i^+$ such that it filled every voxel in the $z=2$ plane not occupied by $p_i$ in raised position. So for a piece to fit it would have to squeeze into the $z=3$ plane, and for the reason stated above no piece meets this requirement.
\end{proof}

\begin{lemma}
	If  $q_i^-$ or $q_i^+$ is placed in the empty bounding box, $p_i$ can be placed subsequently, after which no other pieces can fit into the board.
\end{lemma}
\begin{proof}
	Since  $q_i^-$ and $q_i^+$ span every dimension of the bounding box they only have one valid placement up to rotation and reflection, so this is an immediate consequence of Lemmas~\ref{lemma normal} and~\ref{lemma raised}.
\end{proof}

\begin{lemma}\label{mold blocks blockers}
	Once the $G$-mold is placed, no $q_i^-$ or $q_i^+$ can fit in the board.
\end{lemma}
\begin{proof}
	Note that every $q$ piece fills one of the $z$-planes entirely. Once the $G$-mold is placed, there is at least one voxel filled in every $z$-plane by construction, so it is impossible to place a $q$ piece since there is no empty $z$-plane remaining.
\end{proof}

\subsubsection{Play}
\begin{lemma}
    The first player to move from an initially empty board state has a winning strategy in this instance of Polycube Packing if and only if the second player has a winning strategy in Node Kayles on $G$.
\end{lemma}
\begin{proof}
 First, say Player I has a winning strategy in Node Kayles on $G$ (where they play first). Then they will have a winning strategy in this instance of Polycube Packing where Player II plays first, since if Player II plays a non-$G$-mold piece on the first move, Player I wins immediately by Lemmas~\ref{lemma normal} -~\ref{mold blocks blockers}, and if Player II plays the $G$-mold, then Lemma~\ref{hard with mold} tells us that on the next turn this game is equivalent to Node Kayles on $G$ where Player I goes first, so Player I can force a win. On the other hand, if Player I can't force a win in Kayles when they play first, then in the Polycube Packing game where Player II plays first, on their first turn Player II can play the $G$-mold, and again due to Lemma~\ref{hard with mold} the game becomes equivalent to Kayles on $G$ with Player I moving first and thus Player I cannot force a win.
\end{proof}

\subsection{Partisan Polycube Packing}
In \cite{sch}, Bigraph Node Kayles is also proved PSPACE-complete. This game is identical to Node Kayles except the vertices are partitioned into two independent sets, one for each player, and a player is only allowed to play vertices from their set. It is easy to see that the same exact reduction shows that the partisan version of Polycube Packing is hard in the case where the $G$-mold has already been placed; we construct the same reduction as before, and give each player only the vertex pieces corresponding to their vertices in $G$. To extend to the case of an initially empty box, we will first generate the $q_i^-$ and $q_i^+$ and switch the roles of who plays first, just as before. Then, we put the $G$-mold in the piece set of the first player, and the $q_i^-$ and $q_i^+$ in the piece set of the second player. By Lemmas 3.4-3.7, we see that the first player is forced to play the $G$-mold, at which point the game becomes identical to Bigraph Node Kayles on $G$. So Partisan Polycube Packing is also PSPACE-hard, even from an initially empty box.

\section{2-Player 3D $n$-tris}\label{tetris}
One simple way to convert Tetris into a 2-player packing game is as follows: players alternate taking control of the keyboard in a standard game of Tetris, and whichever player causes the game to end on their turn loses. In this section we will show that if we generalize this game by introducing a third dimension and allowing the pieces to be arbitrarily-sized polycubes, deciding a winner is PSPACE-complete. The main distinctions between this game and 2-Player Polycube Packing are that in this game the pieces have a fixed order of play, and the rules for piece placement are much more restricted. Following \cite{tetris}, we will study the ``offline'' version of 2-Player Tetris, where the entire sequence of pieces is known to all parties ahead of time. The complexity of a rather different 2-player version of Tetris was analyzed in \cite{2ptetris}, in which players have separate boards, and when one player clears a line, their ``garbage'' appears on the other player's screen.

2-Player 3D $n$-tris is defined as follows: the input is an $N \times M \times K$ box, oriented so that one face is considered the ``top'', and a fixed sequence of polycubes $S = (s_1, ..., s_n)$. On turn $i$, the current player places $s_i$ into the box by a continuous downward motion (using any number of rotations) from the top face without intersecting any other pieces. Every piece must be placed with support, meaning its bottom must touch the top of some previously placed piece; another way to phrase this is that there is gravity in the box and pieces fall until they are caught by other pieces. Player 1 moves on odd-numbered turns, Player 2 moves on even-numbered turns, and the first player who cannot move loses. To make this game more similar to traditional Tetris, one could add the feature of plane-clearing: when some $xy$-plane is filled completely, its contents disappear and everything above it falls down by one unit.

\begin{theorem}
Deciding a winner in an instance of 2-Player 3D $n$-tris starting from an initially empty board is PSPACE-complete, with or without plane-clearing.
\end{theorem}

\subsection{Overview of Reduction}
Given an instance of Node Kayles in the form of a graph $G$ with $n$ vertices, we will construct a 2-Player 3D $n$-tris game which will play out in $n$ phases, each phase corresponding to a single turn in Node Kayles. First, we construct a very large piece called the mold which fills most of the box and leaves several cavities open. This will be the first piece in the piece sequence and the first player will be forced to place it with the cavities facing up under threat of immediate loss. After the mold is placed, the phases begin. In the $i^{th}$ phase, the player who moves on turn $i$ in Kayles will be presented with a sequence of $n$ pieces, one for each vertex in $G$. The player will be forced to place one of these pieces into the ``region of contest'', and the rest into the ``phase $i$ dump''. The piece placed into the region of contest represents the vertex they select on that turn in Kayles. In a similar fashion to the 2-Player Polycube Packing reduction, the region of contest will represent the current position of the Kayles game. In order to allow one player to make $n$ moves in a row while the other player does nothing, we will construct a ``garbage chute'', which is a very deep $1 \times 1 \times k$ cavity in the mold, and present the other player with ``garbage pieces'' which are $1 \times 1 \times \frac{k}{n^2}$ columns. If we make $k$ sufficiently large relative to the rest of the construction, a player presented with garbage will be forced to throw it down the garbage chute, and this allows us to effectively skip that player's turn, thus completing the reduction. Finally, to extend hardness to the case with plane-clearing, we will modify the mold so that it is impossible for any plane in the box to ever be full.

\subsection{Hardness Without Plane-Clearing}
Given a graph $G$, we will construct an instance of 2-Player 3D $n$-tris such that the first player has a winning strategy in this game if and only if the first player in Kayles does. Let $n$ denote the number of vertices in $G$ and let $m$ denote the number of edges.

\subsubsection{Phases}
As described above, this game will play out in $n$ phases, each corresponding to a player's turn in Kayles. The odd phases correspond to Player 1's turns in Kayles, and the even phases correspond to Player 2's turns, so we will say that Player 1 (2) is ``in phase'' during the odd (even) phases and ``out of phase'' during the even (odd) phases. The phase that the game is currently in will be fully defined by the next piece in the piece sequence during a certain position of 2-Player 3D $n$-tris; Section~\ref{sequence} specifies how the piece sequence is partitioned into phases.

\subsubsection{Mold: General Structure}
First, we will construct a polycube which fills nearly the entire box, but leaves open several deep cavities on one side, which we will call the ``play side''. We present the general structure of the mold here, and we will specify more precisely the dimensions of the cavities after we construct the other relevant pieces. Essentially, the mold contains a series of $n + 1$ large cavities, plus a small hole in the corner. At the bottom is a cavity we will refer to as the region of play, then above that is a cavity we will call the phase 1 dump, then above that the phase 2 dump, continuing all the way up to the phase $n$ dump at the top. Additionally, at the top left corner there is a narrow and very deep cavity called the garbage chute. Figure~\ref{fig:tetris frame} shows the mold with its play side facing out of the page.

\begin{figure}[H]
	\centering
	\includegraphics[scale=.5]{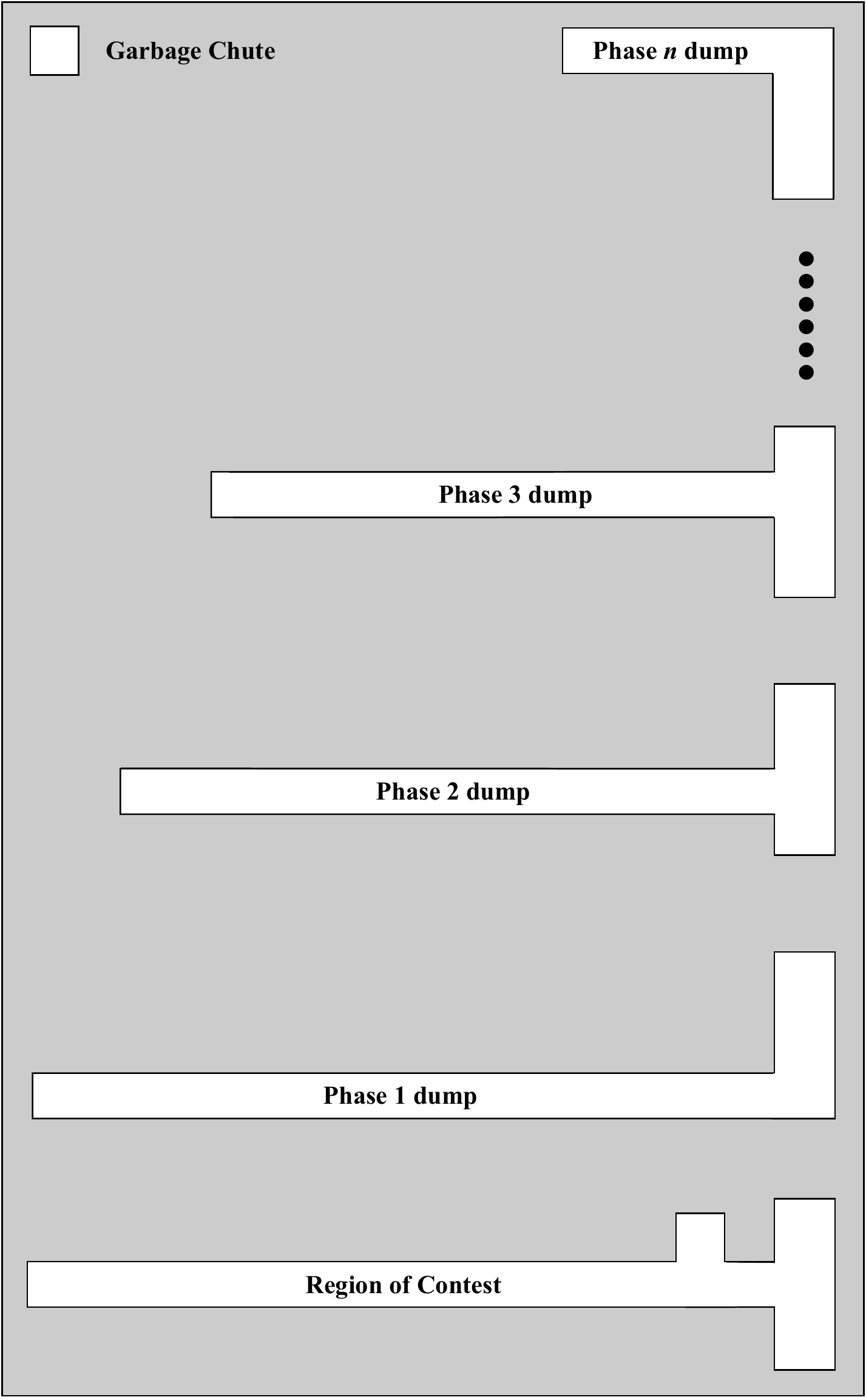}
	\caption{Mold with play side up. Grey region is filled in all the way to the top of the box, while white regions are cavities reaching towards (but not touching) the bottom of the box.}
	\label{fig:tetris frame}
\end{figure}

\subsubsection{Vertex Pieces}
We will construct a vertex piece for each vertex and each phase. For the phase $i$ vertex piece for a vertex $v$, we will construct a polycube which is the union of the following smaller polycubes:
\begin{enumerate}
    \item For each edge $e_h$ incident on $v$, construct an ``edge column'': $[h-1,h] \times [0,1] \times [0,n]$
    \item Construct a ``binding'' which holds the edge columns together: $[0,m] \times [-1,0] \times [i-1, i]$
    \item Construct a ``handle'' which extends off the binding: $[\lceil{\frac{m}{2}}\rceil, \lceil{\frac{m}{2}}\rceil+1] \times [i - (n + 3), -1] \times [i-1,i]$
    \item Construct a ``stub'' which hangs off the side of the handle, one unit away from where the handle meets the binding: $[\lceil{\frac{m}{2}}\rceil-1, \lceil{\frac{m}{2}}\rceil] \times [-3, -2] \times [i-1,i]$. In general we will consider the stub to be a part of the handle.
\end{enumerate}
This construction is shown for an example vertex in Figure~\ref{fig:tetris vertex}.
\begin{figure}[H]
	\centering
	\begin{subfigure}{.475\textwidth}
		\centering
		\includegraphics[width=.8\linewidth]{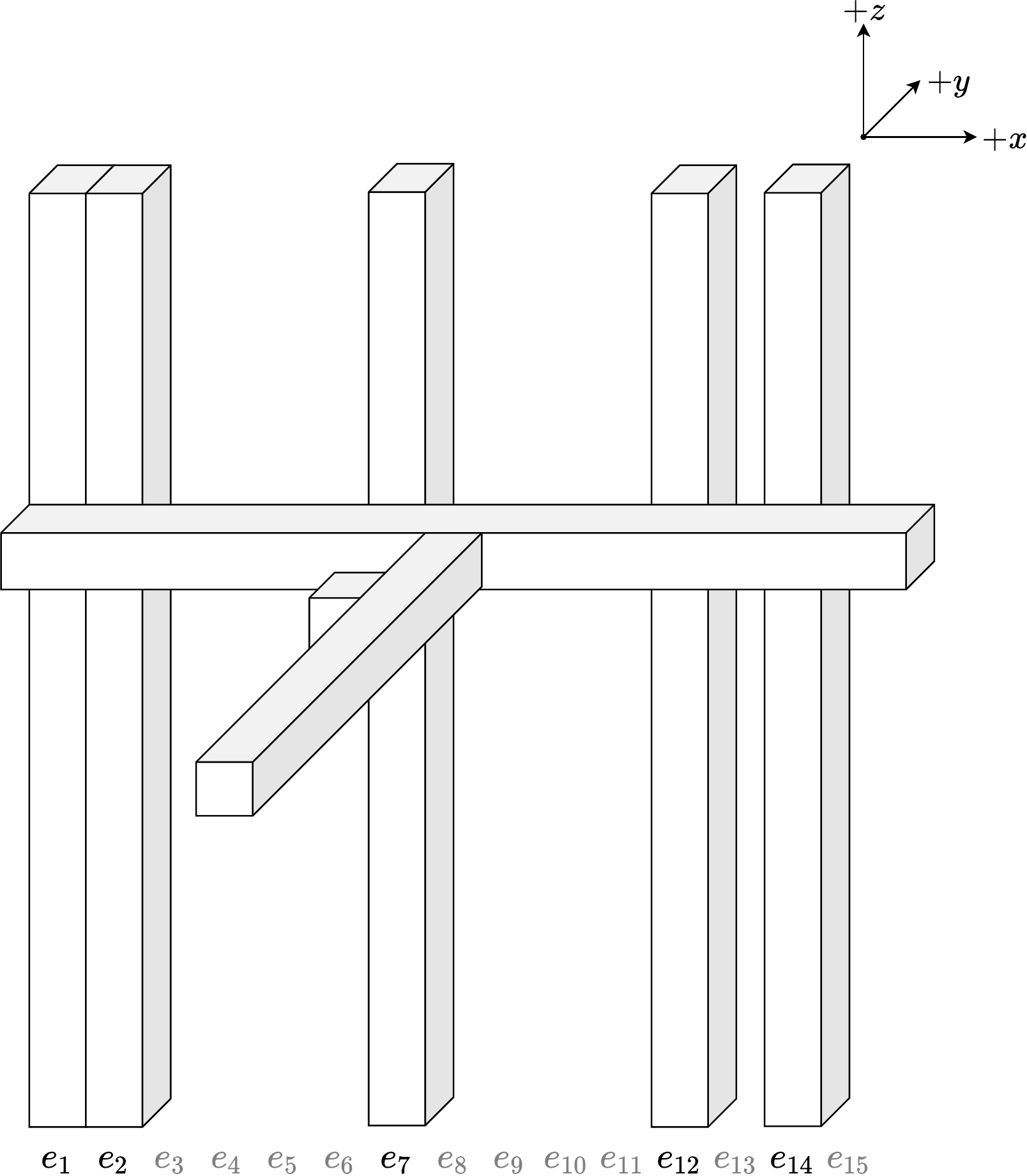}
		\caption{Side view of vertex piece}
		\label{fig:vertex side}
	\end{subfigure}%
		\begin{subfigure}{.525\textwidth}
		\centering
		\includegraphics[width=.8\linewidth]{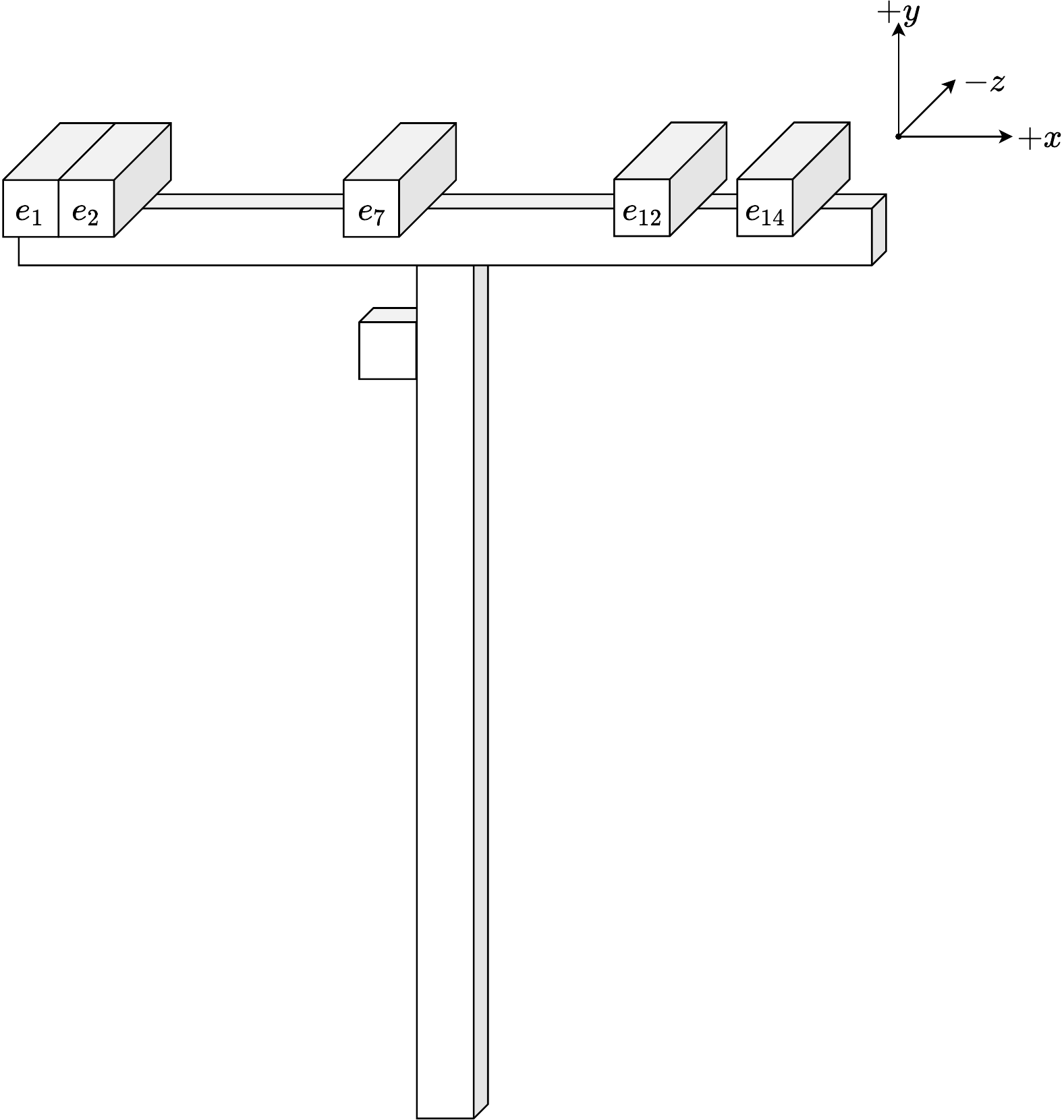}
		\caption{Top-down view of vertex piece}
		\label{fig:vertex top}
	\end{subfigure}
	\caption{Phase 8 vertex piece for a vertex incident on edges $e_1, e_2, e_7, e_{12}, e_{14}$.}
	\label{fig:tetris vertex}
\end{figure}

\subsubsection{Mold: Specific Dimensions}
First we will construct the cavity for the region of contest. Figure~\ref{fig:region of contest} shows the dimensions of an $xy$ cross section of this cavity; extending this along the $z$-dimension by $n$ units gives us the complete construction.

Now we will construct the dumps. For each phase $i$, Figure~\ref{fig:dump} shows the $xy$ cross section of the phase $i$ dump; extending this along the $z$-dimension by $m(n-1)$ units gives us the complete construction.
\begin{figure}[H]
	\centering
	\begin{subfigure}{.5\textwidth}
		\centering
		\includegraphics[width=.8\linewidth]{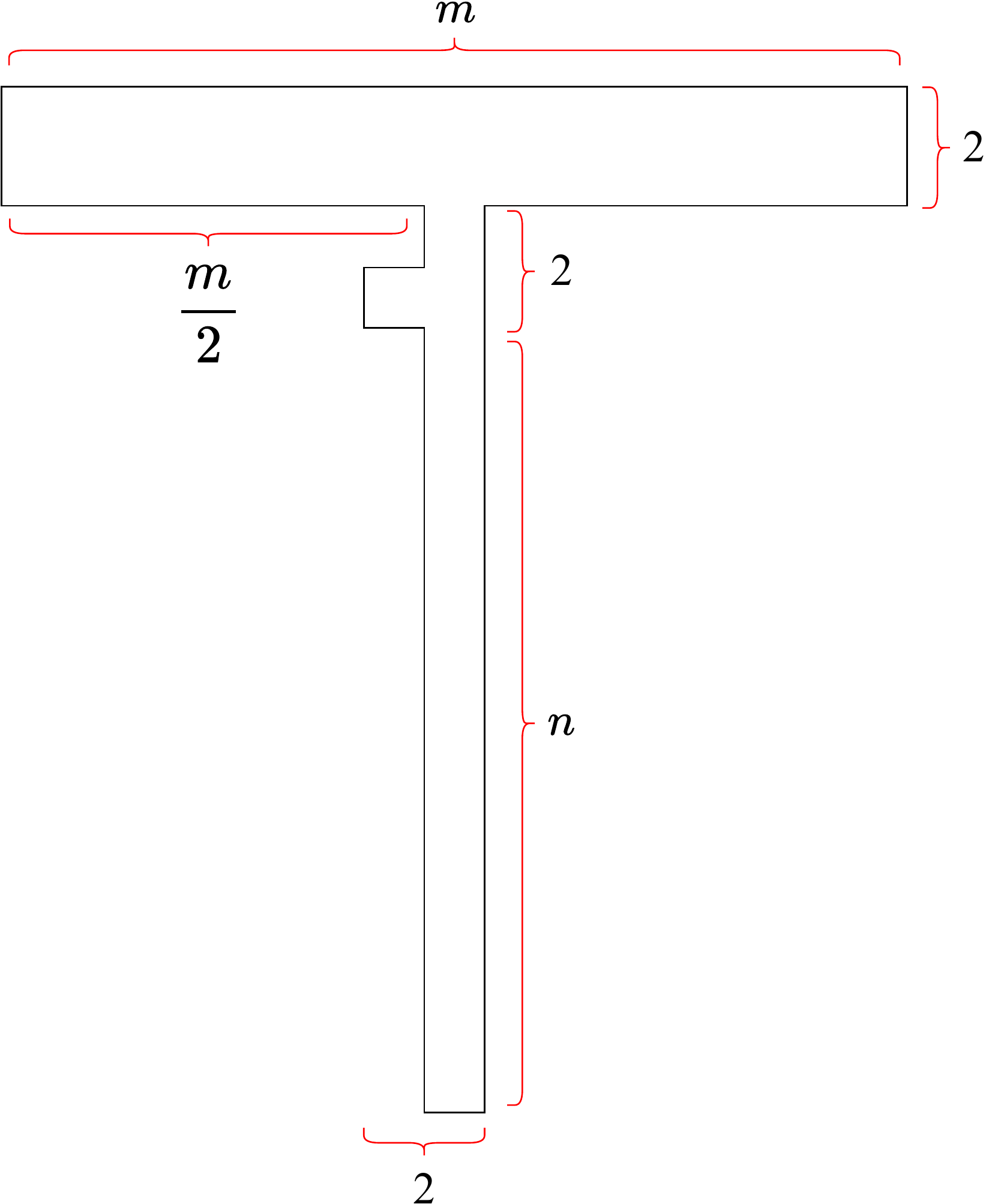}
		\caption{Cross Section of Region of Contest}
		\label{fig:region of contest}
	\end{subfigure}%
		\begin{subfigure}{.5\textwidth}
		\centering
		\includegraphics[width=.8\linewidth]{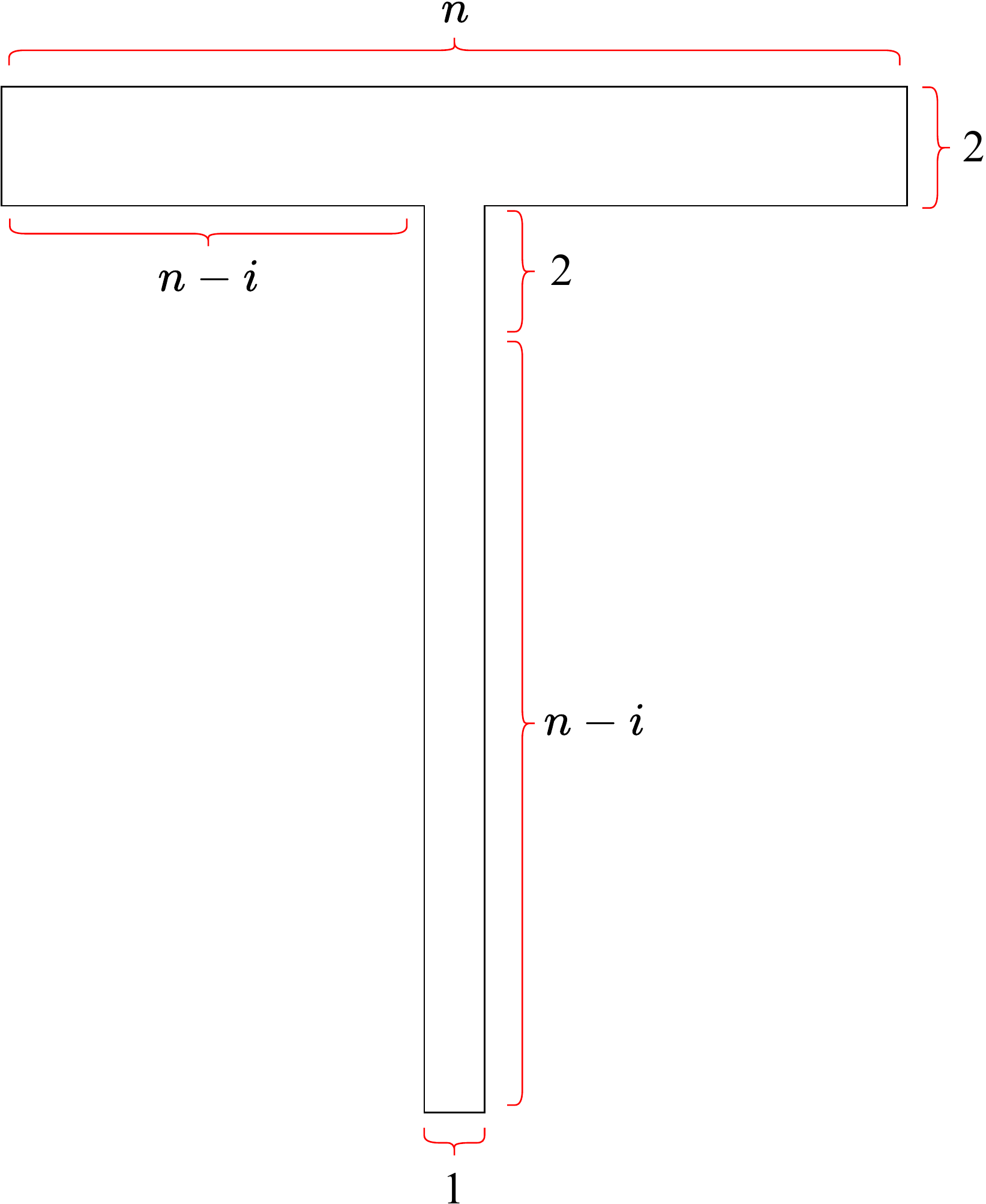}
		\caption{Cross Section of Phase $i$ Dump}
		\label{fig:dump}
	\end{subfigure}
	\caption{Construction of Region of Contest and Dumps}
	\label{fig:contest and dumps}
\end{figure}
Now, arrange these cavities as indentations along the top of a sufficiently large box as shown in Figure~\ref{fig:tetris frame}; the exact arrangement of the various cavities along the surface does not matter, so long as the openings of the cavities are on the same face and they are all pairwise disjoint. This box with various cavities on one side will will be called the ``mold'', and the side of the mold where the cavities open up is the ``play side''. This gives us an almost complete construction of our mold; to finish it we must add a few more very narrow cavities called ``garbage chutes''.

\subsubsection{Garbage and Garbage Chutes}
Take $L$ to be 1 unit greater than maximum length of any dimension over all cavities constructed thus far. Now, extend a solid region from the bottom of the mold along the $z$-dimension until the mold's profile along that dimension is at least $Ln(n-1) + 1$ units long. Next, build an additional cavity that opens up onto the play side whose $xy$ cross section is $1 \times 1$ and which has a depth of $Ln(n-1)$ along the $z$-dimension. This cavity will be referred to as the ``garbage chute'', and an example placement is shown in Figure~\ref{fig:tetris frame}. Now, construct $n(n-1) + 1$ ``garbage pieces'' which are simply $1 \times 1 \times L$ columns. Finally, add one additional garbage chute on each face of the mold which isn't the play side; these cavities will be shallower than the main garbage chute, with dimensions $1 \times 1 \times L$. The dimensions of the mold can be extended slightly if necessary to ensure these extra garbage chutes don't intersect each other or any other previously constructed cavities.

\subsubsection{Piece Sequence} \label{sequence}
We will construct the piece sequence $S$ as follows: for $1 \leq i \leq n$, let $P_i$, the $i^{th}$ phase sequence, consist of all phase $i$ vertex pieces (in any order), with one garbage piece between every vertex piece. Now construct $S$ to contain the following elements in order:
\begin{enumerate}
    \item The mold
    \item One garbage piece
    \item All of the elements of each $P_i$, in ascending order of $i$
\end{enumerate}
\subsubsection{Play}

\begin{lemma}
    The first player is forced to play the mold with its play side facing upwards on the first turn.
\end{lemma}
\begin{proof}
    Since every non-play side of the mold has garbage chutes exactly large enough to accommodate one garbage piece, and the second piece in $S$ is garbage, if the first player plays the mold with a non-play side facing up the second player can fill the chute with garbage after which there is no available space exposed to the top of the box, so the first player then loses.
\end{proof}

From now on, every claim we make will assume a game state in which the mold has been placed with its play side facing upwards.

\begin{lemma}
    If a player is presented with garbage, they must place it in the garbage chute.
\end{lemma}

\begin{proof}
    Since $L$ was taken to be longer than the longest dimension of all non-garbage-chute cavities, and each garbage piece has length $L$, garbage pieces can only fit into the garbage chute. 
\end{proof}

This effectively allows us to skip the turns of the player who is out of phase during a given phase, since if we interweave vertex pieces on the in-phase player's turns with garbage pieces on the out-of-phase player's turns, the out-of-phase player has no choice but to place their pieces into the garbage, and the contents of the garbage chute have no effect on the rest of the board.

\begin{lemma}
    Vertex pieces have a unique placement into the region of contest.
\end{lemma}

\begin{proof}
It is clear that the handle of a vertex pieces is forced to lie in the $1 \times n + 2$ strip of the region's cavity, so it suffices to show that only one of the four placements of the edge column section of the vertex piece into the $m \times 2$ section of the cavity is possible. We can assume without loss of generality that $n < m$; if this is not the case, we can create ``dummy edges'' which are not incident on any vertex, thus padding the dimensions of any piece whose size depends on $m$ but without effecting game-play. This tells us that the edge column region cannot be placed sideways, since the binding has height $m$ so orienting it vertically would have it rise above the depth of the cavity (which is only $n$). The stub on the handle of each vertex piece and the corresponding indentation on one side of the region of contest prevents any vertex piece from being placed upside down.
\end{proof}

\begin{lemma}
    During phase $i$, the player who is in phase must place exactly one vertex piece into the region of contest.
\end{lemma}

\begin{proof}
Say that one phase $i$ vertex piece has been placed in the region of contest, and that a second phase $i$ vertex piece is then placed on top of it. As seen in Figure ~\ref{fig:double placement}, the handle and binding of the second piece must rest on top of the handle and binding of the first piece. By the previous lemma we see both pieces must be placed facing up, and since they are of the same phase, both have the same length of edge columns rising above their bindings. So the tops of the edge columns of the second piece stick out one unit above the tops of the edge columns of the first, so together the two pieces have a vertical height of at least $n + 1$ and thus cannot fit into the region of contest.

\begin{figure}[H]
\centering
\includegraphics[scale=0.5]{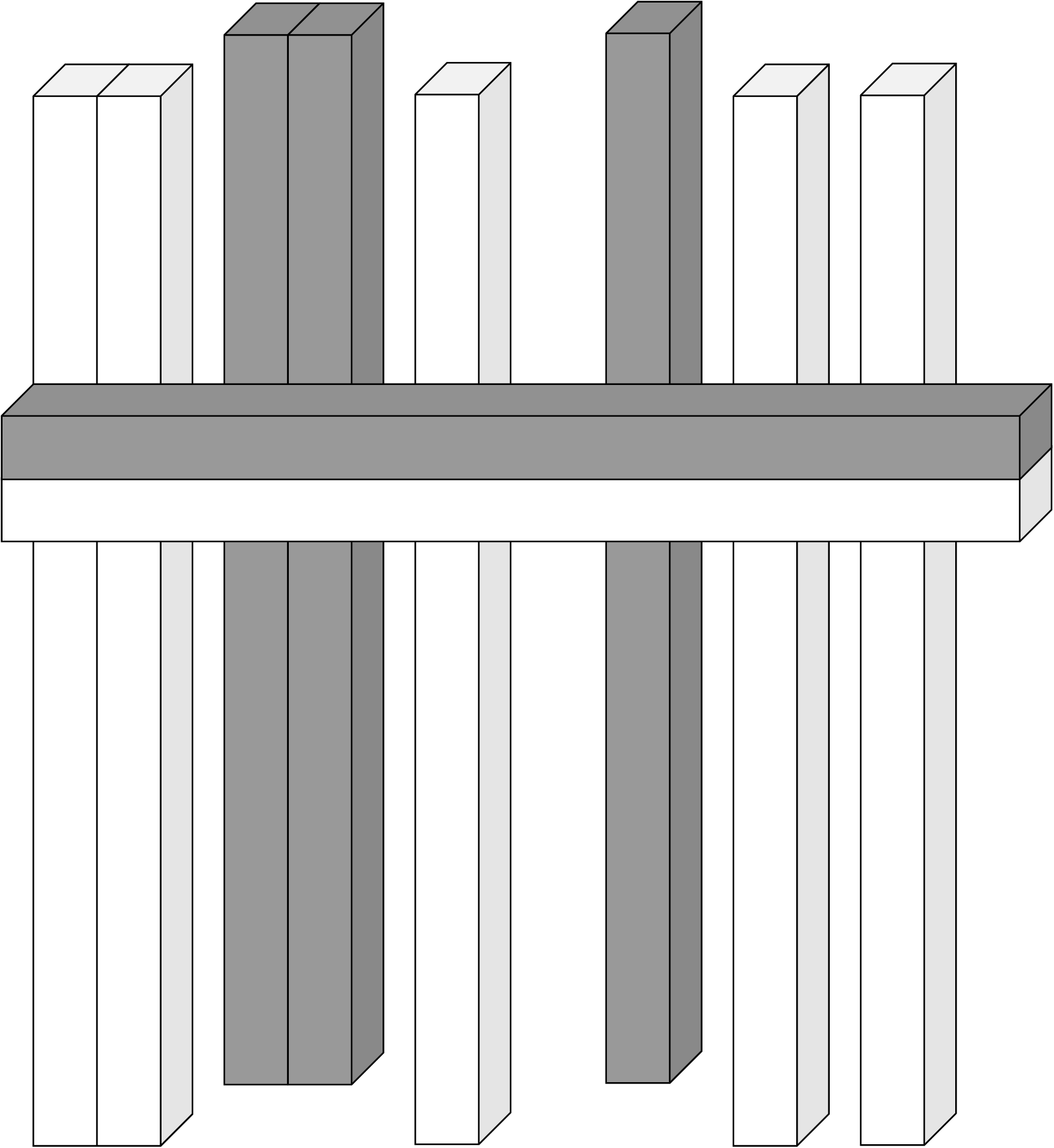}
\caption{Attempting to place two phase 8 vertex pieces into the region of contest. Handles have been removed for clarity. First white was placed, then grey, so grey now sticks out of the top of the cavity.}
\label{fig:double placement}
\end{figure}
Next, we show that for any phase $i$, exactly $n-1$ vertex pieces fit into the phase $i$ dump. Since we showed above that we can assume $n < m$, a vertex piece cannot be placed into a dump in the same orientation as it is placed into the region of contest. This is because the thickest end of the dump has length $n$, and if a vertex piece is placed in the same orientation used for placement in the region of contest, its thickest end will have length $m$ and thus will not fit into the dump. When the first vertex piece is placed sideways into a dump its binding will rise from the floor up to height $m$, and when further pieces are placed, the bindings stack consecutively on top of one another, each with height $m$. Since the dump has depth $m(n-1)$, this means that exactly $n-1$ vertex pieces will fit into the phase $i$ dump.

Now we can prove the main claim. We do so by proving the following stronger statement by induction on $i$: during phase $i$, the player who is in phase must place 1 vertex piece into the region of contest and $n-1$ vertex pieces into the phase $i$ dump. First, note that a phase $i$ vertex piece does not fit into any phase $k$ dump for $k > i$, since the handle of a phase $i$ vertex piece has length $n - i + 2$ while the handle region of the phase $k$ dump is $n - k + 2$ and is thus too short to accommodate the vertex piece. Thus, the base case for our inductive claim is easy: by the above arguments, only one phase 1 vertex piece can fit into the region of contest, and the only available cavity that can fit the rest of the pieces is the phase 1 dump, which fits only $n-1$ pieces, forcing us to put at least one into the region of contest, and the rest into the dump. Now for the inductive case, assume the hypothesis holds up to and including some number $i$; we will now prove it for $i + 1$. As proven above, at most one phase $i + 1$ vertex piece can be placed into the region of contest. Due to our inductive hypothesis, we know that for all $h < i + 1$, the phase $h$ dump contains $n - 1$ phase $h$ vertex pieces and is thus full, and for any $k > i + 1$, the phase $k$ dump is too small to fit a phase $i + 1$ vertex piece. So the only cavity into which the phase $i + 1$ vertex pieces can fit is the phase $i + 1$ dump, and as proven in the previous paragraph, this cavity only fits $n - 1$ pieces. So exactly one phase $i + 1$ vertex piece must be placed into the region of contest, and the other $n - 1$ must be placed into the phase $i + 1$ dump, thus completing the proof.
\end{proof}

\begin{lemma}
    The first player to win in this instance of 2-Player 3D $n$-tris has a winning strategy if and only if the first player in Node Kayles on $G$ does.
\end{lemma}

\begin{proof}
By Lemma 3.5 we know that the strategy of the player who is in phase during phase $i$ is fully defined by which vertex piece they choose to place into the region of contest. By Lemma 3.4 we see that pieces must be placed into the region of contest in their correct upright orientation. In this orientation it is clear that a set of pieces can be placed together in the region of contest in increasing order of phase if and only if they occupy disjoint sets of edge columns, since if two vertices occupy the same edge column they will intersect along that column, and if they do not occupy any common edge columns they will not otherwise intersect by construction. In other words, a vertex piece can be placed into the region of contest if and only if its corresponding vertex does not share an edge with the corresponding vertex of any vertex piece that is already in the region of contest. So choosing which vertex piece to place into the region of contest on phase $i$ corresponds to choosing which vertex in $G$ to mark on turn $i$ in Node Kayles, and a vertex piece corresponding to vertex $v$ can be chosen in a given phase if and only if none of the vertex pieces corresponding to $v$'s neighbors have already been chosen. So we see that this game mimics Node Kayles in its moves and its win conditions, so the first player has a winning strategy in this game if and only if they did in Node Kayles on $G$. 
\end{proof}

\subsection{Hardness With Plane-Clearing}
To extend hardness to the case with plane clearing, we simply modify the non-play sides of the mold to have small indentations to prevent any plane from every being full. Note that in our construction of the mold, each $z$-valued, $xy$-parallel plane has its lower left and right corners filled (see the lower boundary in Figure~\ref{fig:tetris frame}). Now, for every even value of $z$, modify that layer of the mold by removing the cube in the lower left corner, and for every odd value of $z$, modify that layer by removing the cube in the lower right corner. This modification does not effect play since the only new cavities are 1-unit cubes into which no piece fits, and after this modification is made, no $z$-valued plane can ever get filled and clear since one of its lower corners will always be empty. So our reduction holds regardless of the presence of plane-clearing.

\section{Open Problems}\label{open}
\paragraph{2-Player Polycube Packing with Constant Sized Pieces:} An interesting question raised by our Polycube Packing reduction is whether we can establish hardness of polycube packing using pieces of constant complexity, since our reduction relied on the use of highly complex pieces. It seems plausible that this problem is also PSPACE-hard, provided  the bounding box is allowed to be an arbitrary polycube. The analogous 1-player 2D puzzle of optimally packing constant sized polyominoes into a larger polyomino has been proven NP-complete \cite{squarepacking}.
\paragraph{2-Player Polyomino Packing:} Another natural extension of our polycube packing result would be to prove PSPACE-hardness of 2-Player polyomino packing into a rectangle. Key to our reduction was the fact (implicitly proven by our reduction) that every graph can be represented as the intersection graph of a set of polycubes in $\mathbb{R}^3$. The same is not true for polyominoes lying in the plane. Indeed for any arrangement of polyominoes in $\mathbb{R}^2$, we can replace each with a Jordan arc tracing out its interior without effecting pairwise intersections. This implies that polyomino intersection graphs are a subset of the family of ``string graphs'', the graphs representable as the intersection graph of a set of curves in the plane. String graphs properly contain the planar graphs, but beyond that they are not well understood, and their recognition is NP-complete \cite{ss}. In order to extend this proof technique to the planar case, hardness of Node Kayles would thus need to be known for some subset of string graphs, for example planar graphs. Currently, Kayles is not known to be hard on any restricted class of graphs, and in fact admits a polynomial time algorithm for a subset of the string graphs known as cocomparability graphs \cite{nim}.
\paragraph{2-Player Tetris:} Another problem that remains open is settling the complexity of restricted versions of 2-Player 3D $n$-tris, such as 2-Player 2D $n$-tris, 3D $O(1)$-tris (where the pieces are polycubes of constant size), or ideally, 2-Player 2D Tetris.

\section{Acknowledgements}
This paper began as a final project for 6.892, an MIT course on Algorithmic Lower Bounds taught by Erik Demaine. The author would like to thank Erik Demaine along with Jeffrey Bosboom, Adam Hesterberg, and Jayson Lynch for their useful criticisms and commentary on earlier drafts of this paper, and for introducing the author to the study of games and puzzles.
\bibliographystyle{plain}
\include{thebib}
\renewcommand{\thepage}{}
\setlength{\bibsep}{10pt}
\begingroup
    \pagestyle{empty}
    \setlength{\bibsep}{10pt}
    \bibliography{thebib.bib}

\begin{thebibliography}{1}

\bibitem{nim}
Hans Bodlaender and Kratsch Dieter.
\newblock Kayles and nimbers.
\newblock {\em Journal of Algorithms}, 43(1):106--119, 2002.

\bibitem{dem}
Erik~D. Demaine and Martin~L. Demaine.
\newblock Jigsaw puzzles, edge matching, and polyomino packing: Connections and
  complexity.
\newblock {\em Graphs and Combinatorics}, 23(1):195--208, 2007.

\bibitem{totaltetris}
Erik~D. Demaine, Martin~L. Demaine, Sarah Eisenstat, Adam Hesterberg, Andrea
  Lincoln, Jayson Lynch, and Y.~William Yu.
\newblock Total {T}etris: Tetris with monominoes, dominoes, trominoes,
  pentominoes, \ldots.
\newblock {\em Journal of Information Processing}, 25:515--527, 2017.

\bibitem{tetris}
Erik~D. Demaine, Susan Hohenberger, and David Liben-Nowell.
\newblock Tetris is hard, even to approximate.
\newblock {\em Computing and Combinatorics}, pages 351--363, 2003.

\bibitem{squarepacking}
Cristopher Moore and John~Michael Robson.
\newblock Hard tiling problems with simple tiles.
\newblock {\em Discrete and Computational Geometry}, 26(4):573--590, 2001.

\bibitem{2ptetris}
Lev Reyzin.
\newblock 2 player {T}etris is {PSPACE}-hard.
\newblock In {\em Proceedings of the $16^{th}$ Fall Workshop on Computational
  and Combinatorial Geometry, FWCG’06}, 2006.

\bibitem{insanity}
Edward Robertson and Ian Munro.
\newblock {NP}-completeness, puzzles and games.
\newblock {\em Utilitas Math}, pages 99--116, 1978.

\bibitem{ss}
Marcus Schaefer, Eric Sedgwick, and Daniel Stefankovic.
\newblock Recognizing string graphs in {NP}.
\newblock {\em Journal of Computer and System Sciences}, 67(2):365--380, 2003.

\bibitem{sch}
Thomas Schaefer.
\newblock On the complexity of some two-person perfect-information games.
\newblock {\em Journal of Computer and System Sciences}, 16(2):185--225, 1978.

\end{thebibliography}
\endgroup

\end{document}